\documentclass[a4paper]{scrartcl}

\usepackage{amsthm}
\usepackage{amsmath,amssymb}
\usepackage{xspace}
\xspaceaddexceptions{(}

\usepackage{tikz}
\usetikzlibrary{shapes}

\usepackage{fullpage}

\usepackage{rotating}
\usepackage{multirow}
\usepackage{array}

\theoremstyle{plain}
\newtheorem{proposition}{Proposition}
\newtheorem{theorem}{Theorem}
\newtheorem{lemma}{Lemma}
\newtheorem{corollary}{Corollary}

\theoremstyle{definition}
\newtheorem{definition}{Definition}

\newcommand{\myqed}{}

\usepackage{color}

\newcommand{\yes}{\textbf{yes}\xspace}
\newcommand{\no}{\textbf{no}\xspace}

\newcommand{\F}{\mathcal{F}}

\newcommand{\OR}{\textsc{or}\xspace}

\newcommand{\coverILP}{\textsc{Cover ILP}\xspace}
\newcommand{\packingILP}{\textsc{Packing ILP}\xspace}

\newcommand{\independentset}{\textsc{Independent Set}\xspace}

\newcommand{\hittingset}{\textsc{Hitting Set}\xspace}
\newcommand{\ILPF}{\textsc{Integer Linear Program Feasibility}\xspace}

\newcommand{\containment}{\textup{NP}~$\subseteq$~\textup{coNP/poly}\xspace}
\newcommand{\noncontainment}{\textup{NP}~$\nsubseteq$~\textup{coNP/poly}\xspace}

\renewcommand{\P}{\mathcal{P}}
\newcommand{\R}{\mathcal{R}}
\newcommand{\Oh}{\mathcal{O}}

\newcommand{\N}{\ensuremath{\mathbb{N}}\xspace}

\newcommand{\cref}[1]{(\ref{#1})\xspace}

\title{On Polynomial Kernels for Integer Linear Programs: Covering, Packing and Feasibility}
\author{Stefan Kratsch\thanks{Technical University Berlin, Germany, \texttt{stefan.kratsch@tu-berlin.de}}~\thanks{Supported by the DFG, research project PREMOD, KR 4286/1.}}
\date{}

\begin{document}

\maketitle

\begin{abstract}
We study the existence of polynomial kernels for the problem of deciding feasibility of integer linear programs (ILPs), and for finding good solutions for covering and packing ILPs. Our main results are as follows: First, we show that the \textsc{ILP Feasibility} problem admits no polynomial kernelization when parameterized by both the number of variables and the number of constraints, unless \containment. This extends to the restricted cases of bounded variable degree and bounded number of variables per constraint, and to covering and packing ILPs. Second, we give a polynomial kernelization for the \textsc{Cover ILP} problem, asking for a solution to $Ax\geq b$ with $c^Tx\leq k$, parameterized by $k$, when $A$ is row-sparse; this generalizes a known polynomial kernelization for the special case with $0/1$-variables and coefficients ($d$-\textsc{Hitting Set}).
\end{abstract}

\section{Introduction}\label{section:introduction}

This work seeks to extend the theoretical understanding of preprocessing and data reduction for Integer Linear Programs (ILPs). Our motivation lies in the fact that ILPs encompass many important problems, and that ILP solvers, especially CPLEX, are known for their preprocessing to simplify (and shrink) input instances before running the main solving routines (see Atamt\"urk and Savelsbergh~\cite{AtamturkS05} for a survey about modern ILP solvers). When it comes to NP-hard problems, then, formally, being able to reduce \emph{every} instance of some problem would give an efficient algorithm for solving it entirely, and prove~P~$=$~NP (cf.~\cite{HarnikN10}). We avoid this issue by studying the question for efficient preprocessing via the notion of \emph{kernelization} from parameterized complexity~\cite{DowneyF1998_parameterizedcomplexity}, which relates the performance of the data reduction to one or more problem-specific parameters, like the number~$n$ of variables of an ILP.

A kernelization with respect to some parameter~$n$ is an efficient algorithm that given an input instance returns an equivalent instance of size depending only on~$n$; a \emph{polynomial kernelization} guarantees size polynomial in the parameter (see Section~\ref{section:preliminaries} for formal definitions). This notion has been successfully applied to a wide range of problems (see Lokshtanov et al.~\cite{LokshtanovMS12} for a recent survey). A breakthrough result by Bodlaender et al.~\cite{BodlaenderDFH09} (using~\cite{FortnowS11}) gave a framework for ruling~out polynomial kernels for certain problems, assuming \noncontainment~(else the polynomial hierarchy collapses).\footnote{All kernelization lower bounds mentioned in this work are modulo this assumption.}

\textbf{ILP feasibility.} Let us first discuss the \ILPF~(ILPF) problem: Given a set of~$m$ linear (in)equalities in~$n$ variables with integer coefficients, decide whether some integer point~$x\in\mathbb{Z}^n$ fulfills all of them. A well-known result of Lenstra~\cite{Lenstra1983} gives an~$\Oh(\alpha^{n^3}m^c)$ time algorithm for this problem, later improved, e.g., by Kannan~\cite{Kannan87} to~$\Oh(n^{\Oh(n)}m^c)$. We can trivially ``reduce'' to size~$N=\Oh(n^{\Oh(n)})$ by observing that Kannan's algorithm solves all larger instances in polynomial time~$\Oh(Nm^c)=N^{\Oh(1)}$. Can actual reduction rules give smaller kernels, for example with size polynomial in~$n$?

It is clear that we can store an ILPF instance, for example,~$(A,b)$, with~$A\in\mathbb{Z}^{m\times n}$,~$b\in\mathbb{Z}^m$, asking for~$x\geq 0$ with~$Ax\leq b$, by encoding all~$\Oh(nm)$ coefficients, which takes~$\Oh(nm\log C)$ bits where~$C$ is the largest absolute value among coefficients. Let us check what can be said about polynomial kernels with respect to these parameters for ILPF and the~$r$-row-sparse\footnote{Row-sparseness~$r$: at most~$r$ variables per constraint; column-sparseness~$q$: each variable occurs in at most~$q$ constraints; we use~$r$ and~$q$ throughout this work.} variant~$r$-ILPF:

\emph{If the row-sparseness is unrestricted}, then~ILPF($n+C$) and~ILPF($m+C$) encompass \hittingset($n$) and \hittingset($m$)\footnote{\hittingset: Given a base set~$U$ of size~$n$, a set~$\F$ of~$m$~subsets of~$U$, and an integer $k$, find a set of~$k$ elements of~$U$ that intersects each set in~$\F$ (if possible). ILP formulation: Is there~$(x)_{u\in U}$ with~$\sum_{u\in U}x_u\leq k$, and~$\sum_{u\in F}x_u\geq 1$ for all $F\in \F$?}, respectively, which admit no polynomial kernels~\cite{DomLS09,HermelinKSWW11_arxiv}. What about~ILPF($n+m$), which is the maximal open case below the trivial parameter~$n+m+\log C$ (see Figure~\ref{figure:general:ilpf:overview})? 

\emph{With bounded row-sparseness~$r$}, things turn out differently: For~$r$-ILPF($n+C$) and~$r$-ILPF($m+\log C$) there are polynomial kernels: The former is not hard and we briefly explain it in Appendix~\ref{appendix:section:rilpfnC:pk}; the latter is trivial since row-sparseness~$r$ entails~$n\leq r\cdot m$ and hence the above encoding uses~$\Oh(nm\log C)=\Oh(rm^2\log C)=(m+\log C)^{\Oh(1)}$ bits. It was showed previously that~$r$-ILPF($n$) admits no polynomial kernel~\cite{Kratsch13_sparseilp}, and it can be seen that the proof works also for~$r$-ILPF($n+\log C$).\footnote{The used cross-composition with~$t$ input instances creates coefficients of value~$\Oh(t^2)$ with encoding size~$\log C=\Oh(\log t)$ which is permissible for a cross-composition.} 
Again, this leaves parameter~$n+m$ open (see Figure~\ref{figure:rowsparse:ilpf:overview}).

Our contribution for ILPF is the following theorem which, unfortunately, settles both~ILPF($n+m$) and~$r$-ILPF($n+m$) negatively (see Section~\ref{section:ilpfnm:lowerbound}). It can be seen that this completes the picture regarding the existence of polynomial kernels for ILPF and~$r$-ILPF for parameterization by any subset of~$n$,~$m$,~$\log C$, and~$C$. The same is true for the column-sparse case~$q$-ILPF (see Figure~\ref{figure:columnsparse:ilpf:overview}), but we omit a detailed discussion since it is quite similar to the row-sparse case.

\begin{theorem}\label{theorem:main:ilpfnm:nopk}
ILPF($n+m$) does not admit a polynomial kernelization or compression, unless \containment. This holds also if each constraint has at most three variables and each variable is in at most three constraints.
\end{theorem}

It appears that ILPF($n+m$) is the first problem for which we know that a polynomial kernelization fails solely due to the encoding size of large numbers in the input data (and taking into account our proof that no reduction is possible); an additional parameter~$\log C$ would trivially give a polynomial kernel. This of course fits into the picture of hardness results for weight(ed) problems, e.g., W[1]-hardness of \textsc{Small Subset Sum}($k$) where the task is to pick a subset of at most~$k$ numbers to match some target sum~\cite{DowneyF95_completeness} and the kernelization lower bound for \textsc{Small Subset Sum}($k+d$) where the value of numbers is upper bound by~$2^d$~\cite{DomLS09}; however, in both cases the \emph{number} of weights is not bounded in the parameters. Furthermore, there are lower bounds for weighted graph problems (e.g.,~\cite{JansenB11}), but there the used weights have \emph{value} polynomial in the instance size and hence negligible \emph{encoding size}.
We also point out two contrasting positive results: A randomized polynomial compression of \textsc{Subset Sum}($n$)~\cite{HarnikN10}, and a randomized reduction of \textsc{Knapsack}($n$) to many instances of size polynomial in~$n$~\cite{NederlofLZ12} (the number of instances depends on the bit size of the largest weight).

\textbf{Covering and packing ILPs.} Given the overwhelming amount of negative results for ILPF, we turn to the more restricted cases of covering and packing ILPs (cf.~\cite{PlotkinST1991}) with the hope of identifying some positive cases:
\begin{align*}
\mbox{(covering ILP:) }\min\quad & c^Tx & \mbox{\hspace{1cm}(packing ILP:) }\max\quad & c^Tx \mbox{\hspace{2cm}}\\
\mbox{s.t.}\quad & Ax\geq b & \mbox{s.t.}\quad &Ax\leq b\\
&x\geq 0 && x\geq 0
\end{align*}
Here~$A$,~$b$, and~$c$ have non-negative integer entries (coefficients). Feasibility for these ILPs is usually trivial (e.g.,~$x=0$ is feasible for packing), and the more interesting question is whether there exist feasible solutions~$x$ with small (resp.\ large) value of~$c^Tx$. Encompassing many well-studied problems from parameterized complexity, we ask whether~$c^Tx\leq k$ (resp.~$c^Tx\geq k$), and parameterize by~$k$; instances are given as~$(A,b,c,k)$. Unsurprisingly, there are a couple of special cases contained in this setting that have been studied before (e.g., with~$0/1$-variables and coefficients); some of those are W[1]-hard (and are hence unlikely to have polynomial kernels), whereas others have positive results that we could hope to generalize to the more general ILP setting.
To capture the different cases, we use the column-sparseness~$q$ and row-sparseness~$r$ of the matrix~$A$: taking~$q$ and~$r$ as constants, additional parameters, or unrestricted values defines different problems. Our main result in this part is a polynomial kernelization for~$r$-\coverILP($k$) (see Section~\ref{section:rcoverILPk:pk}); the special case of only~$0/1$ variables and coefficients is known as~$r$-\hittingset($k$) and admits kernels of size~$\Oh(k^r)$~\cite{FlumG2006_parameterizedcomplexitytheory,Abu-Khzam10}. Our result also uses the Sunflower Lemma (like~\cite{FlumG2006_parameterizedcomplexitytheory}), but the reduction arguments for sunflowers of linear constraints are of course more involved than for sets.

\begin{theorem}\label{theorem:main:rcoverilpk:pk}
The~$r$-\coverILP($k$) problem admits a reduction to~$\Oh(k^{r^2+r})$ variables and constraints, and a polynomial kernelization.
\end{theorem}

Furthermore, we show how to preprocess instances of~\packingILP($k+q+r)$ and~\coverILP($k+q+r$) to equivalent instances with polynomial in~$kqr$ many variables and constraints (see Section~\ref{section:furtherresults}). For~$r$-\packingILP($k+q$) this is extended to a polynomial kernelization, generalizing that for the special case of bounded degree \independentset($k$). To put these results into context, we provide an overview containing also the inherited hard cases in Tables~\ref{table:packing} and~\ref{table:covering}; for completeness, a brief discussion of these cases can be found in Appendix~\ref{appendix:section:intractable}.

\begin{table}[t]
\centering
\begin{tabular}{|cc||m{3.6cm}|m{3.6cm}|m{3.6cm}|}
\hline
&\multicolumn{4}{c|}{\textbf{Parameterized complexity of Packing ILP($\boldsymbol{k}$)}}\\
\hline
&&\multicolumn{3}{c|}{row-sparseness~$r$}\\
& & \centering constant & \centering parameter & \multicolumn{1}{c|}{unrestricted}\\
\hline
\hline
\multirow{3}{*}{\begin{sideways}column-sparseness~$q$\end{sideways}}& constant&(PK)&(FPT)&\textbf{W[1]-hard from Subset Sum($\boldsymbol{k}$)}\\
\cline{2-5}
& parameter & \textbf{PK (Theorem~\ref{theorem:packingilpkqr:preprocessing})} &\textbf{FPT; $\boldsymbol{n=\Oh(kqr)}$ and $\boldsymbol{m=\Oh(kq^2r)}$ (Theorem~\ref{theorem:packingilpkqr:preprocessing})}&(W[1]-hard)\\
\cline{2-5}
& unrestricted & \textbf{W[1]-hard from Independent Set($\boldsymbol{k}$)}&(W[1]-hard)&(W[1]-hard)\\
\hline
\end{tabular}
\caption{\label{table:packing} ``PK'' stands for polynomial kernel, ``no PK'' stands for no polynomial kernel unless \containment. All normal-font entries are implied by boldface entries.}
\end{table}

\begin{table}[t]
\centering
\begin{tabular}{|cc||m{3.6cm}|m{3.6cm}|m{3.6cm}|}
\hline
&\multicolumn{4}{c|}{\textbf{Parameterized complexity of Cover ILP($\boldsymbol{k}$)}}\\
\hline
&&\multicolumn{3}{c|}{row-sparseness~$r$}\\
& & \centering constant & \centering parameter & \multicolumn{1}{c|}{unrestricted}\\
\hline
\hline
\multirow{3}{*}{\begin{sideways}column-sparsen.~$q$\end{sideways}} & constant &(PK)&(FPT)&\textbf{W[1]-hard from Subset Sum($\boldsymbol{k}$)}\\
\cline{2-5}
& parameter& (PK) & (FPT); \textbf{$\boldsymbol{n=\Oh(kqr)}$ and $\boldsymbol{m=\Oh(kq)}$ (Theorem~\ref{theorem:coveringilpkqr:preprocessing})}&(W[1]-hard)\\
\cline{2-5}
& unrestricted& 
\textbf{PK (Theorem~\ref{theorem:main:rcoverilpk:pk})}&\textbf{FPT but no PK (Proposition~\ref{proposition:covering:hittingsetncase})}&\textbf{W[2]-hard from Hitting Set($\boldsymbol{k}$)}\\
\hline
\end{tabular}
\caption{\label{table:covering} ``PK'' stands for polynomial kernel, ``no PK'' stands for no polynomial kernel unless \containment. All normal-font entries are implied by boldface entries.}
\end{table}


\section{Preliminaries}\label{section:preliminaries}

A \emph{parameterized problem} over some finite alphabet~$\Sigma$ is a language~$\P\subseteq\Sigma^*\times\N$.
The problem~$\P$ is \emph{fixed-parameter tractable} if~$(x,k)\in\P$ can be decided in time~$f(k)\cdot(|x|+k)^{\Oh(1)}$, where~$f$ is an arbitrary computable function.
A \emph{kernelization} for~$\P$ is a polynomial-time algorithm that, given input~$(x,k)$, computes an equivalent instance~$(x',k')$ with~$|x'|,k'\leq h(k)$ where~$h$ is some computable function;~$K$ is a \emph{polynomial} kernelization if~$h(k)$ is polynomially bounded in~$k$. By relaxing the restriction that the created instance~$(x',k')$ must be of the same problem and allow the output to be an instance of any language (i.e., any decision problem) we get the notion of \emph{(polynomial) compression}. Almost all lower bounds for kernelization apply also for this weaker notion.

We also use the concept of an (\OR-)cross-composition of Bodlaender et al.~\cite{BodlaenderJK11} which builds on the breakthrough results of Bodlaender et al.~\cite{BodlaenderDFH09} and Fortnow and Santhanam~\cite{FortnowS11} for proving lower bounds for kernelization.

\begin{definition}[\cite{BodlaenderJK11}] \label{definition:polynomialequivalencerelation}
An equivalence relation~$\R$ on~$\Sigma^*$ is called a \emph{polynomial equivalence relation} if the following two conditions hold:
\begin{enumerate}
\item There is a polynomial-time algorithm that decides whether two strings belong to the same equivalence class (time polynomial in~$|x|+|y|$ for~$x,y\in\Sigma^*)$.
	\item For any finite set~$S \subseteq \Sigma^*$ the equivalence relation~$\R$ partitions the elements of~$S$ into a number of classes that is polynomially bounded in the size of the largest element of~$S$.
\end{enumerate}
\end{definition}

\begin{definition}[\cite{BodlaenderJK11}]\label{definition:crosscomposition}
Let~$L\subseteq\Sigma^*$ be a language, let~$\R$ be a polynomial equivalence relation on~$\Sigma^*$, and let~$\P\subseteq\Sigma^*\times\N$ be a parameterized problem. An \emph{\OR-cross-com\-position of~$L$ into~$\P$} (with respect to $\R$) is an algorithm that, given~$t$ instances~$x_1, x_2, \ldots, x_t \in \Sigma^*$ of~$L$ that are~$\R$-equivalent, takes time polynomial in~$\sum _{i=1}^t |x_i|$ and outputs an instance~$(y,k) \in \Sigma^* \times \mathbb{N}$ such that:
\begin{enumerate}
\item The parameter value~$k$ is polynomially bounded in~$\max_i|x_i|+\log t$.
\item The instance~$(y,k)$ is \yes if and only if \emph{at least one} instance~$x_i$ is \yes.
\end{enumerate}
We then say that~$L$ \OR-cross-composes into~$\P$.
\end{definition}

\begin{theorem}[\cite{BodlaenderJK11}]\label{theorem:orcc}
If an NP-hard language~$L$ \OR-cross-composes into the parameterized problem~$\P$, then~$\P$ does not admit a polynomial kernelization or compression unless \containment and the polynomial hierarchy collapses.
\end{theorem}


\section{A kernel lower bound for ILPs with few coefficients}\label{section:ilpfnm:lowerbound}

In this section, we prove Theorem~\ref{theorem:main:ilpfnm:nopk}, i.e., that \ILPF($n+m$) admits no polynomial kernelization unless \containment.
We begin with a technical lemma that expresses multiplication by powers of two in an ILP. The crucial point is that we need multiplication by~$t$ different powers of two, but can afford only~$\Oh(\log^c t)$ variables and coefficients (direct products of variables are not legal in linear constraints).

\begin{lemma}\label{lemma:power}
Let~$a$,~$b$, and~$p$ be variables, and let~$b_{\max},p_{\max}\geq 0$ be integers. Let~$\ell=\lceil \log p_{\max}\rceil$. There is a system of~$6\ell+7$~linear constraints with~$2\ell-1$~auxiliary variables such that all integer solutions have~$0\leq b\leq b_{\max}$,~$0\leq p\leq p_{\max}$, and~$a=b\cdot 2^p$. Conversely, if these three conditions hold then feasible values for the auxiliary variables exist. The system uses coefficients with bit size~$\Oh(p_{\max}\log b_{\max})$ and  all variables have range at most~$\{0,\ldots,b_{\max}\cdot 2^{p_{\max}}\}$.
\end{lemma}

\begin{proof}
If~$b_{\max}=0$ then constraints~$a=0$,~$b=0$,~$0\leq p$, and~$p\leq p_{\max}$ suffice to prove the lemma; henceforth~$b_{\max}\geq 1$.
Define a constant~$M=b_{\max} \cdot 2^{p_{\max}} +1$; this will be the coefficient with the largest absolute value (bit size as claimed) and also exceed the value of all involved variables. To enforce this and the claimed bounds on~$b$ and~$p$ we add the following constraints.
\begin{align*}
b & \geq 0, & p & \geq 0, & a &\geq 0,\\
b&\leq b_{\max},  & p&\leq p_{\max}, & a &\leq b_{\max}\cdot 2^{p_{\max}} \quad (<M)
\end{align*}
Clearly, we now have~$a,b,p<M$.
Let~$\ell=\lceil \log p_{\max} \rceil$ and add variables~$p_0,\ldots,p_{\ell-1}$ with range~$\{0,1\}$ (this incurs~$2\ell$ constraints) together with a constraint
\begin{align}
p = \sum_{i=0}^{\ell-1} 2^ip_i,
\end{align}
which enforces that the~$p_i$ form a binary encoding of~$p$. Now, we can rewrite~$a=b\cdot 2^p$ as follows (note that this is not added directly as a constraint).
\begin{align*}
a=b\cdot 2^p=b\cdot 2^{\sum_{i=0}^{\ell-1} 2^ip_i}= b \prod_{i=0}^{\ell-1} 2^{2^ip_i} \quad (< M)
\end{align*}
Our strategy is to enforce the partial products~$a_j=b\prod_{i=0}^j 2^{2^ip_i}< M$, for all~$j\in\{0,\ldots,\ell-1\}$, by enforcing~$a_j=2^{2^jp_j}a_{j-1}$; for notational convenience we identify~$a_{\ell-1}:=a$ and~$a_{-1}:=b$ (the remaining~$a_0,\ldots,a_{\ell-2}$ are new auxiliary variables). We add the following constraints for all~$j\in\{0,\ldots,\ell-1\}$:
\begin{align}
a_j &\geq a_{j-1} \label{constraint:aj:lb}\\
a_j &\leq 2^{2^j} a_{j-1} \label{constraint:aj:ub}\\
a_j + M - p_j M &\geq 2^{2^j} a_{j-1} \label{constraint:aj:push}\\
a_j &\leq a_{j-1} + p_j M \label{constraint:aj:pull}
\end{align}
It suffices to show that~$a_j=2^{2^jp_j} a_{j-1}$ for all~$j\in\{0,\ldots,\ell-1\}$. First, observe that constraints~\cref{constraint:aj:lb} and~\cref{constraint:aj:ub} restrict~$a_j$ to range~$\{a_{j-1},\ldots,2^{2^j}a_{j-1}\}$. Second, consider the effect of~$p_j=0$ respectively~$p_j=1$: If~$p_j=0$ then constraint~\cref{constraint:aj:push} is trivially fulfilled since~$M>b_{\max} \cdot 2^{p_{\max}} \geq 2^{2^j} a_{j-1}$, whereas constraint~\cref{constraint:aj:pull} enforces that~$a_j\leq a_{j-1}$; together we get~$a_j=a_{j-1}=2^{2^jp_j}a_{j-1}$ as needed. If instead~$p_j=1$ then constraint~\cref{constraint:aj:pull} is trivially fulfilled whereas constraint~\cref{constraint:aj:push} enforces~$a_j\geq 2^{2^j} a_{j-1}$; together we get~$a_j=2^{2^j} a_{j-1}=2^{2^jp_j} a_{j-1}$ as needed.

Thus, our set of constraints correctly enforces the intended partial products~$a_j$ which implies that~$a=a_{\ell-1}$ has to take the desired value~$b\cdot 2^p$, as claimed.

For the converse, given the above discussion it is easy to check that setting~$p_0,\ldots,p_{\ell-1}$ to the binary expansion of~$p$ and setting~$a_0,\ldots,a_{\ell-2}$ to the values of the corresponding partial products gives a feasible assignment.
\myqed
\end{proof}

Now we are set up to prove the first part of Theorem~\ref{theorem:main:ilpfnm:nopk}.

\begin{lemma}\label{lemma:ilpfnm:nopk}
\ILPF($n+m$) admits no polynomial kernelization or compression unless \containment.
\end{lemma}

\begin{proof}
We give an \OR-cross-composition from the NP-hard \independentset problem. The input instances are of the form~$(G=(V,E),k)$ where~$G$ is a graph and~$k\leq |V|$ is an integer, asking whether~$G$ contains an independent set of size at least~$k$. For the polynomial equivalence relation~$\R$ we let two instances be equivalent if they have the same number of vertices and the same solution size~$k$. It is easy to check that this fulfills Definition~\ref{definition:polynomialequivalencerelation}. For convenience we consider all input graphs to be on vertex set~$V=\{1,\ldots,n\}$, for some integer~$n$.

Let~$t$~$\R$-equivalent instances~$(G_0=(V,E_0),k),\ldots,(G_{t-1}=(V,E_{t-1}),k)$ be given. Without loss of generality we assume that~$t=2^\ell$ for some integer~$\ell$ since otherwise we could copy one instance sufficiently often (at most doubling the input size and not affecting whether at least one of the given instances is \yes).

\textbf{Construction.} We now describe an instance of \ILPF that is \yes if and only if at least one of the instances~$(G_0=(V,E_0),k),\ldots,(G_{t-1}=(V,E_{t-1}),k)$ is \yes for \independentset. We first define an encoding of the~$t$~edge sets~$E_p$ into~$\binom{n}{2}$ integer constants~$D(i,j)$, one for each possible edge~$\{i,j\}\in\binom{V}{2}$ (throughout the proof we use~$1\leq i < j\leq n$):
\begin{align*}
D(i,j):=\sum_{p=0}^{t-1}2^p\cdot D(i,j,p),
&& \mbox{where}
&& D(i,j,p):=\begin{cases}
           1 & \mbox{if~$\{i,j\}\in E_p$,}\\
           0 & \mbox{else.}
          \end{cases}
\end{align*}
In other words, the~$p$-th bit of~$D(i,j)$ is one if and only if~$\{i,j\}$ is an edge of~$G_p=(V,E_p)$. The values~$0\leq D(i,j)\leq 2^t-1$ will be used as constants in the ILP that we construct next.
\begin{enumerate}
 \item We start with a single variable~$p$ that is intended for choosing an instance number; its range is~$\{0,\ldots,t-1\}$. The rest of the ILP will be made in such a way that it ensures that a feasible solution for the ILP will imply that instance~$(G_p,k)$ is \yes for \independentset, and vice versa. \label{step:instancenumber}
 \item Now we will add constraints that allow us to extract the necessary information regarding which of the possible edges~$\{i,j\}$ are present in graph~$G_p$. Recall that the~$p$-th bit of the constant~$D(i,j)$ encodes this. For convenience, we will derive the needed constraints and argue their correctness right away. For each possible edge~$\{i,j\}$ with~$1\leq i<j\leq n$ we introduce a variable~$e_{i,j}$ with the goal of enforcing~$e_{i,j}=1$ if~$\{i,j\}\in E_p$, and~$e_{i,j}=0$ else.
 
 Let~$i,j$ with~$1\leq i < j \leq n$ be fixed (we apply the following for all these choices). Clearly,
 \begin{align*}
  D(i,j)=\sum_{s=0}^{p-1}2^s D(i,j,s) + 2^p D(i,j,p) + \sum_{s=p+1}^{t-1} 2^s D(i,j,s).
 \end{align*}
 We are of course interested in the~$2^p D(i,j,p)$ term, which takes value either~$0$ or~$2^p$. Since the first term~$\sum_{s=0}^{p-1}2^s D(i,j,s)$ is at most~$2^p-1$ and the last is a multiple of~$2^{p+1}$, we will extract it via a constraint
 \begin{align}
 D(i,j)=\alpha+\beta+\gamma,  \label{constraint:edgeextraction}  
 \end{align}
 assuming that we can enforce the following conditions: (i)~$0\leq\alpha\leq 2^p-1$, (ii)~$\beta\in\{0,2^p\}$, and (iii)~$\gamma\in\{0,2^{p+1},2\cdot 2^{p+1},\ldots\}$. (Note that we use new variables~$\alpha_{i,j},\beta_{i,j},\gamma_{i,j},\delta_{i,j},\varepsilon_{i,j}$ for each choice of~$i$ and~$j$, but in the construction the indices are omitted for readability.) 
 This is where Lemma~\ref{lemma:power} comes into the picture as it permits us to enforce the creation of the required values and range restrictions without using overly many variables and constraints. For~(i) we add a variable~$\delta$ and enforce~$\delta=2^p$ by using Lemma~\ref{lemma:power} on~$a=\delta$,~$b=b_{\max}=1$, and~$p$ with~$p_{\max}=t-1$. We then add constraints~$\alpha\geq 0$ and~$\alpha\leq\delta-1$. For~(ii) we apply the lemma on~$a=\beta$,~$b=e_{i,j}$ with~$b_{\max}=1$, and~$p$ with~$p_{\max}=t-1$, enforcing~$\beta=2^pe_{i,j}$. (Note that we want to get~$e_{i,j}=D(i,j,p)$ from~$\beta=2^p\cdot D(i,j,p)$ anyway; this way it is already enforced.) For~(iii) we add a new variable~$\varepsilon\geq 0$ and apply the lemma on~$a=\gamma$,~$b=2\varepsilon$ with~$b_{\max}=2^t-1$ (formally, this requires a new variable~$b'$ and constraint~$b'=2\varepsilon$), and~$p$ with~$p_{\max}=t-1$, enforcing~$\gamma=2^p2\varepsilon=2^{p+1}\varepsilon$. The fact that there are no further restrictions on~$\varepsilon$ effectively allows~$\gamma$ to take on any multiple of~$2^{p+1}$ (and no other values). The upper bound~$b_{\max}$ on~$b=\varepsilon$ for Lemma~\ref{lemma:power} comes from~$D(i,j)\leq 2^t-1$ (we do not require larger values of~$\gamma$ and~$\varepsilon$ since~$\gamma\leq D(i,j)$ and~$\gamma=2^p\varepsilon$).
 
 The most costly application of Lemma~\ref{lemma:power} is for~(iii) where we have~$b_{\max}=2^t-1$ and~$p_{\max}=t-1$. This incurs~$\Oh(\log p_{\max})=\Oh(\log t)$ additional variables and constraints, and uses coefficient bit size~$\Oh(p_{\max}\log b_{\max})=\Oh(t^2)$ (same bounds suffice also for (i) and (ii)). Thus, over all choices of~$1\leq i < j\leq n$, i.e., for getting all the needed edge information, we use~$\Oh(\binom{n}{2}\log t)$ additional variables and constraints. For each variable~$e_{i,j}$ our constraints ensure that it is equal one if~$\{i,j\}$ is an edge in~$G_p$; else it has value zero. \label{step:extractedges}
 \item Now we can finally add the actual edge constraints needed to express the \independentset problem. We add~$n$ variables~$x_1,\ldots,x_n$ with range~$\{0,1\}$, one variable~$x_i$ for each vertex~$i\in V$. For each possible edge~$\{i,j\}\in\binom{V}{2}$, i.e., for all~$1\leq i<j\leq n$ we add the following constraint.
 \begin{align}
 x_i+x_j+e_{i,j}\leq 2  \label{constraint:independence}
 \end{align}
 Finally, we add a constraint
 \begin{align}
 \sum_{i=1}^nx_i\geq k  \label{constraint:target}
 \end{align}
 to ensure that we select at least~$k$ vertices.
 This completes the construction.\label{step:independence}
\end{enumerate} 
Let us now check the number of variables and constraints in our created ILP. The number of variables is dominated by the~$\Oh(\binom{n}{2}\log t)$ variables added in Step~\ref{step:extractedges}, which come from~$\Oh(\binom{n}{2})$ applications of Lemma~\ref{lemma:power}. The same is true for the number of constraints used. Thus both parameters of our target instance are polynomially bounded in the largest input instance plus~$\log t$. Since we postulated no further restrictions on the target ILP (e.g.,~$A$ and~$b$ may have negative coefficients), we will omit a discussion of how to write all constraints as~$Ax\leq b$ with~$x\geq 0$ (here~$x$ is the vector of all variables used) since that is straightforward.
The largest bit size of a coefficient is~$\Oh(t^2)$ hence it is easy to see that the whole ILP can be generated in time polynomial in the total input size (which is roughly~$\Oh(t\cdot n^2)$ from the~$t$ \independentset instances with~$n$ vertices each).
It remains to argue correctness of the construction. 

\textbf{Soundness.} Assume that the created ILP has at least one feasible solution and pick according values for all variables. The goal is to show that instance number~$p$ is \yes for \independentset. To this end, let~$S\subseteq V$ denote the set of vertices~$i$ for which~$x_i=1$. By~\cref{constraint:target} we know that~$|S|\geq k$. Now, for any two vertices~$i,j\in S$ with~$i<j$ it suffices to show that~$\{i,j\}$ is not an edge of the graph~$G_p$. By~\cref{constraint:independence} and~$i,j\in S$ we know that~$e_{i,j}=0$ since~$x_i+x_j+e_{i,j}\leq 2$.
As we already showed in the construction that~$e_{i,j}$ must equal the~$p$-th bit of~$D(i,j)$ this implies that~$\{i,j\}\notin E_p$, as needed.

\textbf{Completeness.} Now, assume that for some~$p^*\in\{0,\ldots,t-1\}$ the instance $(G_{p^*}=(V,E_{p^*}),k)$ is \yes for \independentset and let~$S\subseteq V$ be some independent set of size~$k$ in~$G_{p^*}$. Set~$p=p^*$ and assign the following values to the other variables to get a feasible solution for the ILP:
\begin{itemize}
 \item Set the variables~$x_i$ to one for all~$i\in S$ and to zero for~$i\in V\setminus S$. Clearly this fulfills~\cref{constraint:target}.
 \item Set each edge variable~$e_{i,j}$ to one if and only if~$\{i,j\}$ is an edge of~$G_p$ (and zero else). It is easy to see that together with the~$x$-variables this fulfills all constraints~\cref{constraint:independence} since we only have~$x_i=x_j=1$ if~$\{i,j\}$ is not an edge in~$G_p$ (by choice of~$S$).
 \item For each choice of~$1\leq i<j\leq n$, set~$\delta_{i,j}=2^p$ and set
 \begin{align*}
 \alpha_{i,j}&=\sum_{s=0}^{p-1}2^s D(i,j,s), \\ 
 \beta_{i,j}&=2^p D(i,j,p),\\
 \gamma_{i,j}&=\sum_{s=p+1}^{t-1}2^s D(i,j,s).
 \end{align*}
 Thus~$0\leq \alpha_{i,j}\leq 2^p-1= \delta_{i,j}-1$ and~$\beta=2^pe_{i,j}$. Also, clearly, our choices for~$\alpha$,~$\beta$, and~$\gamma$ fulfill~\cref{constraint:edgeextraction}.
 
 As~$\gamma_{i,j}$ is a multiple of~$2^{p+1}$ we can set~$\varepsilon_{i,j}=\frac{\gamma_{i,j}}{2^{p+1}}$ to agree with~$\gamma_{i,j}=2^{p+1}\varepsilon_{i,j}$ that we enforced via Lemma~\ref{lemma:power}. Similarly,~$\delta_{i,j}=2^p$ and~$\beta_{i,j}=2^pe_{i,j}$ agree with what we enforced by the lemma. Thus, for all three applications of the lemma (per possible edge~$\{i,j\}$) it is guaranteed (by the lemma statement) that feasible values for all auxiliary variables can be found.
\end{itemize}
It follows that the constructed ILP instance has a feasible solution, as required. This completes the correctness part.

Thus we have an \OR-cross-composition from the NP-hard \independentset problem to the \ILPF($n+m$) problem. By Theorem~\ref{theorem:orcc}, this implies that \ILPF($n+m$) has no polynomial kernel or compression unless \containment.
\myqed
\end{proof}

The second part of Theorem~\ref{theorem:main:ilpfnm:nopk} is now an easy corollary. Similarly, we get lower bounds for covering and packing ILP with parameter~$n+m$; the proofs are postponed to Appendix~\ref{appendix:section:corollary:ilpfnm:sparse:nopk} and~\ref{appendix:section:corollary:pilpnm:nopk}. Getting the sparseness for Corollary~\ref{corollary:ilpfnm:sparse:nopk} is easier than what was needed for~$r$-ILPF($n$) in~\cite{Kratsch13_sparseilp} since both number of constraints and number of variables are bounded in the parameter value.

\begin{corollary}\label{corollary:ilpfnm:sparse:nopk}
\ILPF($n+m$) restricted to instances that have at most~$3$ variables per constraint (row-sparseness) and with each variable occurring in at most~$3$ constraints (column-sparseness) admits no polynomial kernel or compression unless \containment.
\end{corollary}

\begin{corollary}\label{corollary:pilpnm:cilpnm:nopk}
\coverILP($n+m$) and \packingILP($n+m$) do not admit polynomial kernelizations or compressions unless \containment.
\end{corollary}


\section{Polynomial kernelization for row-sparse Cover ILP($\boldsymbol{k}$)}\label{section:rcoverILPk:pk}

In this section, we prove Theorem~\ref{theorem:main:rcoverilpk:pk}, by giving a polynomial kernelization for~$r$-\coverILP($k$), generalizing polynomial kernelizations for~$r$-\textsc{Hitting Set}($k$) \cite{FlumG2006_parameterizedcomplexitytheory,Abu-Khzam10}. Our result uses the sunflower lemma (stated below) that can also be used for~$r$-\textsc{Hitting Set} (as in~\cite{FlumG2006_parameterizedcomplexitytheory}). However, the application is complicated by the fact that the replacement rules for a sunflower of constraints are not as clear-cut as for sets in the hitting set case: Constraints forming a sunflower pairwise overlap on the same variables but with (in general) different coefficients; hence no small replacement is implied. Additionally, we have to bound the number of constraints that have exactly the same set of variables with nonzero coefficients, called \emph{scope}, since the sunflower lemma will only be applied to the set of different scopes. The main work lies in the proof of the following lemma; the polynomial kernelization is given as a corollary.

\begin{lemma}\label{lemma:rcoverILPk:preprocessing}
The~$r$-\coverILP($k$) problem admits a polynomial-time preprocessing to an equivalent instance with~$\Oh(k^{r^2+r})$ constraints and variables.
\end{lemma}

Before we turn to the proof, we give a lemma that captures some initial reduction arguments including a bound on the number of constraints with the same scope.

\begin{lemma}\label{lemma:rcoverilpk:basicreduction}
Given an instance~$(A,b,c,k)$ of~$r$-\coverILP($k$) we can in polynomial time reduce to an equivalent instance~$(A',b',c',k)$ such that:
\begin{enumerate}
\item No constraint is satisfied if all variables in its scope are zero, i.e.,~$b'_i\geq 1$.
\item The cost function~$c'^T$ is restricted to~$1\leq c'_i\leq k$.
\item All feasible solutions with~$c'^Tx\leq k$ have~$x_i\in\{0,\ldots,k\}$ for all~$i$.
\item There are at most~$(k+1)^d$ constraints for any scope of at most~$d$ variables.
\end{enumerate}
\end{lemma}

\begin{proof}
Recall that we are considering covering constraints of the form~$A[i,\cdot]x\geq b[i]$.
Constraints that are satisfied when all their variables are zero, are satisfied by all~$x\geq 0$ and may be discarded (as we already require~$x\geq 0$). If~$c_i=0$ for some variable~$x_i$ then we can set~$x_i$ to arbitrarily large values and thus satisfy all constraints containing it (here we use monotonicity of covering constraints); we may thus safely delete any such variable and all constraints containing it.
Similarly, if~$c_i>k$ for some~$x_i$ then already when setting~$x_i\geq 1$ we cannot get cost~$c^Tx$ at most~$k$; thus we delete such variables (but not their constraints).
Consequently, we get~$1\leq c_i\leq k$ for all remaining variables~$x_i$. This also limits the possible values with~$c^Tx\leq k$ to~$x_i\in\{0,\ldots,k\}$. 

Finally, consider any set of~$d$ variables with more than~$(k+1)^d$ constraints having exactly this scope. It is clear that there are only~$(k+1)^d$ possible assignments to the variables that do not violate the maximum cost of~$k$. Each constraint can rule out some of those. It is therefore sufficient to keep only one constraint for each infeasible assignment, as all further constraints are redundant and may be deleted, giving the claimed bound. Note that each constraint has at most~$r$ variables in its scope, hence we can perform this reduction in time polynomial in~$n$.
\myqed
\end{proof}

We recall sunflowers and the sunflower lemma of Erd\H{o}s and Rado~\cite{ErdosR1960}.

\begin{definition}[\cite{ErdosR1960}]\label{definition:sunflower}
Let~$\F$ denote a family of sets. A \emph{sunflower} in~$\F$ of cardinality~$t$ and with \emph{core}~$C$ is a collection of~$t$~sets~$\{F_1,\ldots,F_t\}\subseteq \F$ such that~$F_i\cap F_j=C$ for all~$i\neq j$. The sets~$F_1\setminus C,\ldots,F_t\setminus C$ are called the \emph{petals} of the sunflower; they are pairwise disjoint. The core~$C$ may be empty.
\end{definition}

\begin{lemma}[Sunflower Lemma~\cite{ErdosR1960}]\label{lemma:sunflowerlemma}
Let~$\F$ denote a family of sets each of size~$d$. If the cardinality of~$\F$ is greater than~$d!\cdot(t-1)^d$ then~$\F$ contains a sunflower of cardinality~$t$, and such a sunflower can be found in polynomial~time.
\end{lemma}

\begin{proof}[Proof of Lemma~\ref{lemma:rcoverILPk:preprocessing}]
To begin with, we apply Lemma~\ref{lemma:rcoverilpk:basicreduction} in polynomial time. Afterwards, for each constraint scope with~$d\leq r$ variables, there are at most~$(k+1)^d=\Oh(k^r)$ constraints with that scope.
Now, we will apply the sunflower lemma to the set of all scopes of size~$d$, for each~$d\in\{1,\ldots, r\}$. If there are more than~$d!\cdot(t-1)^d=\Oh(t^r)$ sets then we find a sunflower consisting of~$t$ sets (scopes). We will show how to remove at least one constraint matching one of the scopes, when~$t$ is sufficiently large (and polynomially bounded in~$k$).

If, instead, the number of scopes is at most~$d!\cdot(t-1)^d$ (for all~$d$) then we can bound the total number of constraints as follows: We have~$r$ choices for~$d\in\{1,\ldots, r\}$. For each~$d$ we have at most~$d!\cdot(t-1)^d=\Oh(t^r)$ constraint scopes. For each scope there are at most~$\Oh(k^r)$ constraints. In total this gives a bound of~$r\cdot \Oh(t^r) \cdot \Oh(k^r) = \Oh(t^r\cdot k^r)$.
Since each constraint has at most~$r$ variables we get the same bound (in~$\Oh$-notation) for the number of variables.

\textbf{Removing a constraint.} Let us now see how to find a redundant constraint when there are more than~$d!\cdot(t-1)^d$ constraint scopes for some~$1\leq d\leq r$. We will also derive an appropriate value for~$t$ (at least~$t>k$).
Consider a~$t$-sunflower in the set of constraint scopes of size~$d$ (which we can get in polynomial time from the sunflower lemma). Let its core be denoted by~$C=\{x_1,\ldots,x_s\}$, with~$0\leq s<d\leq r$, and its pairwise disjoint petals by~$\{y_{1,s+1},\ldots,y_{1,d}\},\ldots,\{y_{t,s+1},\ldots,y_{t,d}\}$.
(Note that~$s<d$ is needed since all sets are different, which requires nonempty petals; else they would all equal the core.)
Thus there must be constraints in the ILP matching these scopes. We arbitrarily pick one constraint for each scope:
\begin{align*}
a_{1,1}x_1+a_{1,2}x_2+\ldots+a_{1,s}x_s+a_{1,s+1}y_{1,s+1}+\ldots+a_{1,d}y_{1,d}&\geq b_1\\
a_{2,1}x_1+a_{2,2}x_2+\ldots+a_{2,s}x_s+a_{2,s+1}y_{2,s+1}+\ldots+a_{2,d}y_{2,d}&\geq b_2\\
\vdots\quad\quad\quad\quad\quad\quad\quad\quad\quad\quad\quad\quad&\\
a_{t,1}x_1+a_{t,2}x_2+\ldots+a_{t,s}x_s+a_{t,s+1}y_{t,s+1}+\ldots+a_{t,d}y_{t,d}&\geq b_t
\end{align*}
Note that to keep notation simple the indexing of the variables and coefficients is only with respect to the sunflower and makes no assumption about the actual numbering within~$Ax\geq b$; all arguments are local to the sunflower.

First, let us note that we may return \no~(or a dummy \no-instance) if the core is empty and~$t>k$: That would give us more than~$k$ constraints on disjoint sets of variables that each require at least one nonzero variable; this is impossible at maximum cost~$k$. In the following, assume that~$s\geq 1$.

Since each variable takes values from~$\{0,1,\ldots,k\}$ (by Lemma~\ref{lemma:rcoverilpk:basicreduction}), there are at most~$(k+1)^s$ possible assignments for the~$s$ core variables (in fact there are even less since the sum is at most~$k$). It is clear that assigning zero to all core variables does not lead to a feasible solution since each of the~$t$ constraints requires at least one nonzero variable and the petal variables are disjoint. However, unlike for~$r$-\textsc{Hitting Set}, assigning one to a single core variable is not necessarily sufficient to satisfy all constraints, and the value of each variable for a constraint might be quite different. Thus, we may not simply replace the constraints of the sunflower by the restriction of one constraint to the core variables.

To cope with this difficulty we employ a marking strategy. We check all~$(k+1)^s$ assignments of choosing a value from~$\{0,1,\ldots,k\}$ for each core variable. It is possible that some of the constraints are already satisfied by this core assignment due to the monotone behavior of covering constraints. If more than~$k$ constraints need an additional nonzero variable (which would be a petal variable) then clearly this core assignment is infeasible. In this case we arbitrarily mark~$k+1$ constraints that are not satisfied by the core assignment alone (i.e., if all their other variables would be zero); these serve to prevent the core assignment from being chosen. If at most~$k$ constraints are not yet satisfied then we mark all of them. Clearly, in total we mark at most~$(k+1)\cdot (k+1)^s$ constraints.

We will now argue that all unmarked constraints can be deleted. Clearly, deletion can only create false positives, so consider a solution to the instance obtained by deleting any unmarked constraint. If the assignment does not also satisfy the removed constraint (where we take value zero for variables that do not occur after deletion) then, in particular, the same is true for the core assignment made by this assignment. Hence, while marking constraints with respect to this core assignment we must have marked~$k+1$ other constraints (that are also not satisfied by the core assignment alone), which are hence not deleted. However, these other constraints cannot all be satisfied with a budget of at most~$k$, a contradiction. Thus, deleting all unmarked constraints is safe.

Therefore, if we have a sunflower with core size~$s$ and more than~$(k+1)\cdot (k+1)^s$ constraints then our marking procedure allows us to delete at least one constraint. Thus, allowing for core size~$s$ up to~$r-1$, we set~$t=(k+1)\cdot(k+1)^{r-1}+1=\Oh(k^r)$. While we can find~$t$-sunflowers (via the sunflower lemma) we can always delete at least one constraint. 
Our earlier discussion at the start of the proof now gives a bound of~$\Oh(t^r\cdot k^r) = \Oh(k^{r^2+r})$
on the number of constraints and variables achieved by this reduction process. This completes the proof.
\myqed
\end{proof}

\begin{corollary}\label{corollary:rcoverILPk:pk}
The~$r$-\coverILP($k$) problem admits a polynomial compression to size~$\Oh(k^{r^2+2r})$ and a polynomial kernelization. 
\end{corollary}

\begin{proof}[Proof (sketch)]
We know how to reduce to~$\Oh(k^{r^2+r})$ constraints in polynomial time. Since each constraint can equivalently be described by the infeasible assignments that it defines for the variables in its scope, we may encode the instance by replacing each constraint on~$d\leq r$ variables by a~$0/1$-table of dimension~$(k+1)^d$: Each~$0$ entry means that the assignment corresponding to this coordinate is infeasible (e.g., entry at~$(2,3,5)$ tells whether~$x_1=2$,~$x_2=3$, and~$x_3=5$ is feasible for this constraint, when~$(x_1,x_2,x_3)$ is its scope); each~$1$ stands for a feasible entry. For each constraint this requires~$(k+1)^d=\Oh(k^r)$~bits for the table plus~$r\cdot \log(k^{r^2+r})=\Oh(\log k)$~bits for encoding the scope (there are at most~$r$ variables, each one encoded by a number from~$1$ to~$\Oh(k^{r^2+r})$). The cost function~$c^Tx$ requires only the storing of~$\Oh(k^{r^2+r})$ integers with values from~$0$ to~$k$, taking~$\Oh(\log k)$ bits each. In total this gives the claimed size for the compression.

To get a polynomial kernelization let us first observe that the problem described above is clearly in NP since feasible solutions (certificates) can be verified in polynomial time. Thus, by an argument of Bodlaender et al.~\cite{BodlaenderTY11}, we get an encoding as an instance of~$r$-\coverILP($k$) by using the implicit Karp reduction whose existence follows from the fact that~$r$-\coverILP($k$) is complete for NP. The size of this instance is polynomial in~$\Oh(k^{r^2+2r})$, and hence polynomial in~$k$. This completes the polynomial kernelization.
\myqed
\end{proof}


\section{Further results}\label{section:furtherresults}

In this section we state some results for \packingILP($k$) and \coverILP($k$) with column-sparseness~$q$ and row-sparseness~$r$ as additional parameters or constants; the proofs are postponed to Appendix~\ref{appendix:section:proofs:theorem:packingilpkqr:preprocessing} and~\ref{appendix:section:proofs:theorem:coveringilpkqr:preprocessing}.

\begin{theorem}\label{theorem:packingilpkqr:preprocessing}
\packingILP($k+q+r$) is fixed-parameter tractable and admits a polynomial-time preprocessing to~$\Oh(kqr)$ variables and~$\Oh(kq^2r)$ constraints. For~$r$-\packingILP($k+q$) this gives a polynomial kernelization.
\end{theorem}

\begin{theorem}\label{theorem:coveringilpkqr:preprocessing}
\coverILP($k+q+r$) is fixed-parameter tractable and admits a polynomial-time preprocessing to~$\Oh(kqr)$ variables and~$\Oh(kq)$ constraints. (A polynomial kernelization for~$r$-\coverILP($k+q$) follows from Corollary~\ref{corollary:rcoverILPk:pk}.)
\end{theorem}


\section{Conclusion}\label{section:conclusion}

We have studied different problem settings on integer linear programs.
For \ILPF parameterized by the numbers~$n$ of variables and~$m$ of constraints we ruled out polynomial kernelizations, assuming \noncontainment. This still holds when both column- and row-sparseness are at most three. Adding further new and old results, this settles the existence of polynomial kernels for ILPF and~$r$-ILPF for parameterization by any subset of~$\{n,m,\log C, C\}$ where~$C$ is the maximum absolute value of coefficients.

Regarding covering and packing ILPs, we gave polynomial kernels for~$r$-\coverILP($k$), generalizing~$r$-\hittingset($k$), and for~$r$-\packingILP($k+q$).
Further results and observations give an almost complete picture regarding~$q$ and~$r$ except for the question about polynomial kernels for parameter~$k+q+r$.
We recall that for both problems one can reduce to~$n,m\leq(kqr)^{\Oh(1)}$ but parameterization by~$n$ and~$m$ only admits no polynomial kernelization (Corollary~\ref{corollary:pilpnm:cilpnm:nopk}).


\bibliographystyle{plain}
\bibliography{main}


\newpage
\appendix

\section{Omitted proofs from Section~\ref{section:ilpfnm:lowerbound}}


\subsection{Proof of Corollary~\ref{corollary:ilpfnm:sparse:nopk}} \label{appendix:section:corollary:ilpfnm:sparse:nopk}

\begin{proof}[Proof (sketch)]
It suffices to show how an instance of \ILPF($n+m$) can be modified to fulfill the extra spareness restrictions, while not blowing up the numbers of variants and constraints to more than~$(n+m)^{\Oh(1)}$ (formally this constitutes a polynomial parameter transformation from the unrestricted to the sparse case). We make no particular assumption about the constraints other than demanding that they are expressible as~$\sum_{i=1}^n\alpha_ix_i \blacktriangleleft b$ where the~$x_i$ are the variables,~$\alpha_i$'s and~$b$ are integer coefficients, and~$\blacktriangleleft\in\{\leq,\geq,=\}$. By replacing the sums with partial sums~$s_j=\sum_{i=1}^j\alpha_ix_i$ computed via constraints~$s_j=s_{j-1}+\alpha_j x_j$ (or, formally,~$s_j-s_{j-1}-\alpha_jx_j=0$), we can reduce the number of variables per constraint to three. Per original constraint this requires~$\Oh(n)$ new variables for partial sums as well as~$\Oh(n)$ constraints, for a total of~$\Oh(nm)$ additional variables and constraints. To reduce also the number of occurrences of each variable, it suffices to make copies of the variables. Conveniently, the variables for the partial sums occur only two times each (intuitively, once for forcing their own value and once for enforcing the next value, resp., for the occurrence of~$s_n$ in~$s_n\blacktriangleleft b$). Each original variable occurs at most once in each of the~$m$ original constraints. When moving to the partial sums, this is replaced by at most one occurrence for computing the partial sum needed for a given constraint; thus it suffices to have~$m$ copies~$x_i(1),\ldots,x_i(m)$ to use in the partial sum constraints. Clearly, this can be achieved via
\begin{align*}
x_i=x_i(1) && x_i(1)=x_i(2) && \ldots && x_i(m-1)=x_i(m).
\end{align*}
This gives two occurrences per variable~$x_i(j)$ with the third being in place of~$x_i$ in a partial sum constraint. We use an additional~$\Oh(nm)$ variables and~$\Oh(nm)$ further constraints for this part.

It is easy to see that we do not affect feasibility of the ILP since we only make equivalent replacements, respectively introduce copies of variables. In total we use~$\Oh(nm)$ additional variables and constraints, meaning that the final row- and column-sparse ILP has polynomial in~$n+m$ many variables and constraints as needed. If we would have a polynomial kernel or compression for this problem, then combined with the present reduction this would constitute a polynomial compression for the \ILPF($n+m$) problem that we started with, implying \containment. This completes the proof.
\myqed
\end{proof}


\subsection{Proof of Corollary~\ref{corollary:pilpnm:cilpnm:nopk}}\label{appendix:section:corollary:pilpnm:nopk}

\begin{proof}
We prove the corollary for the case of \coverILP($n+m$); the case of packing ILPs works along the same lines. To this end, we start from the instance created in the proof of Lemma~\ref{lemma:ilpfnm:nopk} and turn it into a covering ILP plus a specified bound on the maximum target function value.  The main properties that we require is that each variable in the created ILP (for which we ask for feasibility) has an upper bound on its range, e.g.,~$p\leq t-1$ and~$\beta_{i,j}\leq 2^t$, and that all variables are non-negative. For convenience let us write~$B(z)$ for the upper bound on the range of any variable~$z$ (irrespective of the actual names used in the proof of Lemma~\ref{lemma:ilpfnm:nopk}).

The instance created in the proof of Lemma~\ref{lemma:ilpfnm:nopk} asks whether a given set of linear constraints admits a feasible point with integer coordinates. For all involved variables~$z$ we have~$0\leq z \leq B(z)$ for some a priori known bound~$B(z)$. Let~$n$ denote the number of variables, and let~$z_1,\ldots,z_n$ be an enumeration of the variables. Each constraint can be written as
\begin{align*}
\sum_{i=1}^n\alpha_i z_i \blacktriangleleft b,
\end{align*}
where~$\blacktriangleleft\in\{\leq,\geq,=\}$ and where~$b$ and each~$\alpha_i$ are integer constants. Obviously, each equality constraint can be split into one~$\leq$ and one~$\geq$ constraint (with the same coefficients). Subsequently, we may turn each~$\leq$ constraint into a~$\geq$ constraint be multiplying all its coefficients by~$-1$.

It remains to handle the fact that we may have negative coefficients. This is where the bounds on the variables come in. For each variable~$z_i$ with upper bound~$B(z_i)$ we add a constraint
\begin{align}
z_i+\hat{z}_i\geq B(z_i), \label{constraint:equality:lower}
\end{align}
where~$\hat{z}_i$ is a new variable, constrained to non-negative values by~$\hat{z}_i\geq 0$. To ensure that we actually get~$z_i+\hat{z}_i= B(z_i)$, we let our target function be
\begin{align}
\min \quad \sum_{i=1}^n z_i + \sum_{i=1}^n \hat{z}_i.
\end{align}
By asking for a feasible solution of cost (target function value) at most~$\sum_{i=1}^n B(z_i)$ we ensure that all constraints~\cref{constraint:equality:lower} must be fulfilled with equality (these constraints enforce a lower bound of~$\sum_{i=1}^n B(z_i)$ on the cost). Thus, for all~$i$, we have~$z_i+\hat{z}_i=B(z_i)$. This permits us to replace all occurrences of~$z_i$ with negative coefficient, say~$-\gamma$ where~$\gamma\geq 1$, by plugging in~$z_i=B(z_i)-\hat{z}_i$. We get~$-\gamma z_i= - \gamma B(z_i) + \gamma \hat{z}_i$; moving~$- \gamma B(z_i)$ into the RHS of the inequality (i.e., adding~$\gamma B(z_i)$ to both sides) we have thus reduced the number of negative coefficients by one. Afterwards, should there be a negative RHS coefficient, then we may return a dummy \no instance. (Note however that we only make equivalent transformations so if there was a feasible point for the initial ILP then this translates naturally to a feasible solution for the obtained covering ILP.)

When all this is done, we have only~$\geq$ constraints with non-negative coefficients, and all variables are restricted to non-negative values. Thus we have obtained an equivalent covering ILP with~$2n$ variables and at most~$n+2m$ constraints, which transfers the kernelization lower bound from Lemma~\ref{lemma:ilpfnm:nopk} also to \coverILP($n+m$). The proof for \packingILP($n+m$) is analogous.
\myqed
\end{proof}


\section{Omitted Proofs from Section~\ref{section:furtherresults}}


\subsection{Proof of Theorem~\ref{theorem:packingilpkqr:preprocessing}}\label{appendix:section:proofs:theorem:packingilpkqr:preprocessing}

We give two lemmas that together imply Theorem~\ref{theorem:packingilpkqr:preprocessing}.

\begin{lemma}\label{lemma:packingilpkqr:basicreduction}
Let an instance~$(A,b,c,k)$ of \packingILP($k$) be given and let~$q$ and~$r$ denote the column- and row-sparseness of~$A$. In polynomial time we can compute an equivalent instance~$(A',b',c',k)$ such that
\begin{enumerate}
\item $A'$ has dimension at most~$kq^2r\times kqr$, i.e., the new ILP has at most~$kq^2r$ constraints and at most~$kqr$ variables.
\item $1 \leq c_i \leq k-1$, i.e., each increase in a variable contributes a value of at least one and at most~$k-1$.
\item $A'$ is a submatrix of~$A$, and~$b'$ and~$c'$ are subvectors of~$b$ and~$c$ respectively (this means also that column- and row-sparseness of~$A'$ are at most the same as for~$A$).
\end{enumerate}
\end{lemma}

\begin{proof}
Let us begin with a few simple observations: First, we may delete all variables~$x_i$ with~$c_i=0$ since we have no increase in value~$c^T x$ for setting~$x_i$ to a nonzero value but it may make constraints harder to fulfill; it is crucial that all constraints are of the form~$A[i,\cdot]x\leq b_i$. Second, we check for variables~$x_i$ such that setting~$x_i=1$ already makes some constraint infeasible, i.e., the assignment defined by
\begin{align*}
x_j=\begin{cases}
     1 & \mbox{if~$j=i$,}\\
     0 & \mbox{else,}
    \end{cases}
\end{align*}
is infeasible. We may then delete~$x_i$ since it cannot contribute to~$c^Tx$ (for feasible~$x$). Finally, if~$c_i\geq k$ for any~$x_i$ then we may answer \yes since we already know that we can set~$x_i$ to one (all others to zero) which gives~$c^Tx\geq k$. Henceforth~$1\leq c \leq k-1$.

Now, let us see how to get the upper bounds of~$kq^2r$ and~$kqr$ on the number of constraints and variables. We greedily select variables such that no two appear in a shared constraint; we end up with a set~$S$ of variables. If~$|S|\geq k$ then assigning one to each variable in~$S$ is feasible (since each constraint ``sees'' only one variable with value one), and the pay off~$c_ix_i$ per variable~$x_i$ is at least one, so the value of this is at least~$k$. Thus, if~$|S|\geq k$ then we may answer \yes. Otherwise, if~$|S|< k$ then every variable not in~$S$ must be in a shared constraint with some variable~$x_i\in S$. Each such constraint has at most~$r-1$ other variables beyond~$x_i$ and each~$x_i\in S$ appears in at most~$q$ constraints. Thus the total number of variables is at most
\begin{align*}
|S|+(r-1)q|S|\leq kqr.
\end{align*}
Since each variable appears in at most~$q$ constraints this also limits the number of constraints to at most~$kq^2r$. (If some constraint has~$b_i>0$ but no variables with nonzero coefficients then of course the answer is \no instead.) Since we only delete variables and constraints from the original instance it is clear that, e.g.,~$A'$ is a submatrix of~$A$.
\myqed
\end{proof}

According to the lemma we can always reduce to at most~$kq^2r$ constraints and at most~$kqr$ variables. Thus, for~\packingILP($k+q+r$), i.e., parameterized by~$k$,~$q$, and~$r$, this constitutes a preprocessing to an equivalent instance with a polynomial number of coefficients. This also means that the problem is fixed-parameter tractable by an application of Lenstra's~\cite{Lenstra1983} algorithm for bounded dimension ILPs. When~$r$ is constant then this reduction gives a polynomial kernelization, as we show next.

\begin{lemma}\label{lemma:rpackingilpkq:pk}
$r$-\packingILP($k+q$) admits a polynomial kernelization.
\end{lemma}

\begin{proof}
To begin, we apply the reduction given by Lemma~\ref{lemma:packingilpkqr:basicreduction}.
With the goal of finding a feasible point~$x$ with~$c^Tx\geq k$ it is clear that it suffices to use a maximum range of~$\{0,\ldots,k\}$ for each variable~$x_i$ (if~$x$ is feasible then so are all~$x'$ with~$0\leq x'\leq x$, and a single position of value~$k$ suffices). Furthermore, as for~$r$-\coverILP($k$), any constraint on~$r$ variables excludes some of the~$|\{0,\ldots,k\}|^r$ possible assignments. Thus we can follow the same argumentation as in the proof of Corollary~\ref{corollary:rcoverILPk:pk} for encoding the feasible points: For each of the~$kq^2r$ constraints, we encode in~$\Oh(r\log(kqr))=\Oh(\log(kq))$ the names of its~$r$ variables, followed by a binary array of bit size~$|\{0,\ldots,k\}|^r=\Oh(k^r)$ that encodes for all points in~$\{0,\ldots,k\}^r$ whether they are feasible (according to this constraint). Finally, we add an encoding of the target function, i.e., a sequence of~$kqr$ values from~$\{0,\ldots,k-1\}$ which takes~$\Oh(kq\log k)$. This gives a compression to bit size
\begin{align*}
\Oh(kq^2r\cdot (k^r+\log(kq))+kq\log k)=\Oh(k^{r+1}q^2+kq^2 \log(kq))=(k+q)^{\Oh(1)}.
\end{align*}

Again, it can be observed that the encoded problem is in NP, as a candidate solution~$x$ can be efficiently verified (solution value and consistency with all local lists of feasible points). Thus, by a Karp reduction back to~$r$-\packingILP($k+q$) we get a polynomial kernelization, as the output size is polynomial in the compression obtained above.
\myqed
\end{proof}


\subsection{Proof of Theorem~\ref{theorem:coveringilpkqr:preprocessing}}\label{appendix:section:proofs:theorem:coveringilpkqr:preprocessing}

Now, we will address Theorem~\ref{theorem:coveringilpkqr:preprocessing}. The arguments are analogous to those needed for Theorem~\ref{theorem:packingilpkqr:preprocessing} so we will focus on the differences.

\begin{lemma}\label{lemma:coverilpkqr:basicreduction}
Let an instance~$(A,b,c,k)$ of \coverILP($k$) be given and let~$q$ and~$r$ denote the column- and row-sparseness of~$A$. In polynomial time we can compute either an equivalent instance~$(A',b',c',k)$ such that
\begin{enumerate}
\item $A'$ has dimension at most~$kq\times kqr$, i.e., the new ILP has at most~$kq$ constraints and~$kqr$ variables.
\item $1\leq c \leq k$, i.e., each increase in a variable contributes a cost of at least one and at most~$k$.
\item $A'$ is a submatrix of~$A$, and~$b'$ and~$c'$ are subvectors of~$b$ and~$c$ respectively (this means also that column- and row-sparseness of~$A'$ is at most the same as for~$A$).
\end{enumerate}
\end{lemma}

\begin{proof}
The reduction to~$c\geq 1$ and~$b\leq 1$ follows analog to Lemma~\ref{lemma:rcoverilpk:basicreduction}. If~$c_i>k$ for any variable~$x_i$ then we cannot afford to set it to one (or larger) and may delete it from the ILP. Now, each remaining variable~$x_i$ that is set to~$x_i\geq 1$ costs at least one and each remaining constraint requires such a variable. Since we have a total budget of~$k$ and each variable appears in at most~$q$ constraints this limits the number of constraints to at most~$kq$ (else we return \no). The number of variables is at most~$kqr$ since each constraint contains at most~$r$ variables. Clearly, all required operations take only polynomial time and the produced matrix~$A'$ and vectors~$b'$ and~$c'$ are submatrices respectively subvectors of~$A$,~$b$, and~$c$.
\myqed
\end{proof}

Thus the \coverILP($k+q+r$) problem can be reduced to at most~$kq$ constraints and at most~$kqr$ variables, i.e., to a polynomial number of coefficients. Regarding polynomial kernels for~$r$-\coverILP($k+q$) we refer to our more general result for~$r$-\coverILP($k$) by Lemma~\ref{lemma:rcoverILPk:preprocessing} and Corollary~\ref{corollary:rcoverILPk:pk}. Nevertheless, if the column-sparseness~$q$ is indeed small, it is of course possible to obtain a polynomial compression and kernelization via Lemma~\ref{lemma:coverilpkqr:basicreduction} combined with a direct analog of Lemma~\ref{lemma:rpackingilpkq:pk} (the encoding is entirely the same with slightly better bounds due to fewer constraints and variables; the interpretation of the encoded problem is of course as a covering problem).


\section{Basic intractable cases for covering and packing ILPs}\label{appendix:section:intractable}

In this section we briefly discuss the hardness results that can be observed from various hard problems being captured as packing or covering ILP problems. 
The first two are immediate; we spell out the ILP for W[1]-hardness of~$r$-\packingILP($k$) for completeness, but omit a detailed discussion.

\begin{proposition}\label{proposition:independentset}
The~$r$-\packingILP($k$) problem is hard for W[1] for all~$r\geq 2$ by reduction from \independentset($k$).
\end{proposition}

\begin{proof}
Given an instance~$(G=(V,E),k)$ for \independentset($k$), make an indicator variable~$x_v$ for each vertex~$v\in V$, and create the following ILP:
\begin{align*}
\max\quad & \mathbf{1}^Tx\\ 
\mbox{s.t.}\quad& x_u+x_v\leq 1 \quad \forall\{u,v\}\in E\\
&x\geq 0
\end{align*}
where~$\mathbf{1}$ denotes the all one vector of dimension~$|V|$. The created \packingILP($k$) instance asks for a feasible point~$x$ with value~$\mathbf{1}^Tx$ at least~$k$.
\myqed
\end{proof}

\begin{proposition}\label{proposition:hittingset}
The \coverILP($k$) problem is hard for W[2] by reduction from \hittingset($k$).
\end{proposition}

If we restrict the column-sparseness to any constant~$q\geq 4$ then both \coverILP($k$) and \packingILP($k$) remain W[1]-hard.

\begin{proposition}\label{proposition:ccolumnsparse:w1hard}
Let~$q\geq 4$. The~$q$-\coverILP($k$) and the~$q$-\packingILP($k$) problem are W[1]-hard.
\end{proposition}

\begin{proof}
For both problems we give a reduction from \textsc{Subset Sum}($k$) which is W[1]-hard due to Downey and Fellows~\cite{DowneyF95_completeness} (see also~\cite{DowneyF1998_parameterizedcomplexity}). In \textsc{Subset Sum}($k$) we are given~$n$ non-negative integer values, say,~$a_1,\ldots,a_n$, along with a target value~$t$ and the question is whether there exists a choice of at most~$k$ numbers whose sum is exactly~$t$; the problem is parameterized by~$k$.
To simplify the reduction we ask for a subset with \emph{exactly}~$k$ numbers (this can be ensured by adding~$k$ times the value~$0$ for padding any subset matching the target sum to cardinality~$k$). We also assume that~$a_i\leq t$ for all~$i$ since larger values can be discarded.

We begin with the reduction to~$q$-\packingILP($k$). We create the following packing ILP with variables~$x_1,\ldots,x_n$, asking for a feasible point~$x$ with~$\mathbf{1}^Tx\geq k$.
\begin{align*}
\max\quad & \mathbf{1}^Tx\\
\mbox{s.t.}\quad & \sum_{i=1}^nx_i\leq k && x\leq 1 \\
& \sum_{i=1}^n a_ix_i\leq t && \sum_{i=1}^n (t-a_i)x_i\leq (k-1)t\\
& x\geq 0
\end{align*}
Clearly, each variable appears in exactly~$4$ constraints, apart from~$x\geq 0$. (If variables are defined to be~$0/1$, or if range constraints~$x_i\in\{0,1\}$ are free, then we get~$q\geq 3$.)
Note that the following equivalence holds; it uses that feasible solutions with value~$\sum_{i=1}^nx_i\geq k$ satisfy also~$\sum_{i=1}^nx_i=k$.
\begin{align*}
\sum_{i=1}^n (t-a_i)x_i\leq (k-1)t &\quad\Leftrightarrow\quad t\sum_{i=1}^n x_i - \sum_{i=1}^n a_ix_i \leq (k-1)t\\
& \quad\Leftrightarrow \quad kt-\sum_{i=1}^n a_ix_i \leq (k-1)t\\
& \quad\Leftrightarrow \quad \sum_{i=1}^n a_ix_i \geq t
\end{align*}
Thus feasible solutions with value at least~$k$ satisfy also~$\sum_{i=1}^nx_i=k$ and $\sum_{i=1}^n a_ix_i=t$. Since each~$x_i$ is either zero or one this translates directly to a choice of~$k$ numbers~$a_i$ whose sum is exactly~$t$. The reverse holds easily, by setting the~$x_i$ according to a choice of~$k$ numbers with sum exactly~$t$.

For~$q$-\coverILP($k$) the same style of ILP works, except for a caveat: We do not have the ability to enforce~$x\leq 1$. While this could be replaced by
\begin{align*}
\sum_{i\neq j}x_i\geq k-1
\end{align*}
for all~$j\in\{1,\ldots,n\}$, this would violate the desired bound on the column-sparseness. If the problem setting permits to define a range for each variable~$x_i$ then of course we can add~$x_i\in\{0,1\}$ instead; this would give column-sparseness $q=3$.
Otherwise, using a derandomized form of color-coding (e.g., using splitters~\cite{NaorSS95}), we could first show W[1]-hardness for Multicolored Subset Sum($k$) where the numbers are colored~$1,\ldots,k$ and the task is to match the target sum by picking one number of each color (the back reduction to Subset Sum($k$) is easy by encoding the colors into the numbers). Then we could augment our reduction by demanding that the sum over all indicator variables for numbers of the same color be at least one (giving column-sparseness~$q=4$). This naturally enforces that each variable takes value at most one, since we have exactly~$k$ colors. Nevertheless, since this is not the main focus of this work, we will not go into the details of this. For completeness we state the covering ILP with~$0/1$-variables.
\begin{align*}
\min\quad & \sum_{i=1}^n x_i\\
\mbox{s.t.}\quad & \sum_{i=1}^nx_i\geq k && x_i\in\{0,1\}\\
& \sum_{i=1}^n a_ix_i\geq t && \sum_{i=1}^n (t-a_i)x_i\geq (k-1)t 
\end{align*}
Correctness can be argued analogously to the case of~$q$-\packingILP($k$).
\myqed
\end{proof}

Finally, for \coverILP($k+r$), i.e., parameterized by target solution cost plus row-sparseness, it is straightforward to show fixed-parameter tractability but the problem does not admit a polynomial kernelization by encompassing \textsc{Hitting Set}($k+d$), parameterized by solution size and maximum set size~$d$; for the latter problem Dom et al.~\cite[Theorem~6]{DomLS09} ruled out polynomial kernels, assuming \noncontainment.

\begin{proposition}\label{proposition:covering:hittingsetncase}
The \coverILP($k+r$) problem is fixed-parameter tractable but admits no polynomial kernel or compression unless \containment.
\end{proposition}

\begin{proof}
Let us first sketch an~$\Oh(r^k(n+m)^{\Oh(1)})$ time fpt-algorithm following similar algorithms for, e.g., $r$-\textsc{Hitting Set}($k$)~\cite{DowneyF1998_parameterizedcomplexity}. We may assume that the cost function~$c^Tx$ has~$c_i\geq 1$ for all~$i\in\{1,\ldots,n\}$ (otherwise such variables~$x_i$ can be removed along with their constraints, as setting them to arbitrarily large values satisfies all their constraints but costs zero). The algorithm starts from~$x=0$ and tries to find a minimal feasible solution~$x$ with~$c^Tx\leq k$ by branching. In each recursive call with some partial solution~$x$, it picks an arbitrary infeasible constraint and branches on the choice of one of its at most~$r$ variables, and increases the value of that variable by one (in~$x$). This increases the cost~$c^Tx$ by at least one. In each branch, this is repeated until either a feasible solution is found, or until we have reached the maximum budget of~$k$. It is easy to prove correctness by using the invariant that each level of the branching tree contains at least one (possibly infeasible) solution~$x$ with~$x\leq x^*$ (position-wise), where~$x^*$ is an arbitrary feasible solution with cost~$c^Tx$ at most~$k$ (if one exists). The runtime is clearly dominated by the~$\Oh(r^k)$ leaves of the search tree (times polynomial factors). 

The reduction from \textsc{Hitting Set}($k+d$) to \coverILP($k+r$) is straightforward: Make an indicator variable for each element of the base set and ask for a solution with sum of indicators at most~$k$. For each set, demand that the sum of indicators for its elements is at least~$1$; this gives row-sparseness~$r=d$.
\myqed
\end{proof}


\section{A polynomial kernelization for $r$-ILP Feasibility($n+C$)}\label{appendix:section:rilpfnC:pk}

\begin{proposition}\label{proposition:rilpfnC:pk}
The $r$-\ILPF($n+C$) problem admits a polynomial kernelization.
\end{proposition}

\begin{proof}
The key fact is that there can be only few different constraints in an instance of~$r$-ILPF($n+C$): Any linear constraint on~$d\leq r$ variables, say,
\begin{align*}
 \alpha_1 x_{i_1} + \alpha_2 x_{i_2} + \ldots + \alpha_d x_{i_d} \blacktriangleleft \beta
\end{align*}
is characterized by the following aspects.
\begin{enumerate}
 \item The choice of relation symbol~$\blacktriangleleft\in\{\leq,\geq,=\}$.
 \item The choice of~$d$ variables:~$\binom{n}{d}=\Oh(n^d)$.
 \item The values of~$\alpha_1,\ldots,\alpha_d$ and~$\beta$. By problem restriction, the absolute values of the coefficients are at most~$C$, hence for each one we have~$2C+1$ choices out of~$\{-C,\ldots,0,\ldots,C\}$. Thus, there are~$(2C+1)^{d+1}$ choices.
\end{enumerate}
Over all~$d\leq r$ this gives a total number of~$m=\Oh(n^rC^{r+1})$ different constraints. Thus, discarding all duplicate constraints (same variables and coefficients) and using the canonical encoding in~$\Oh(nm\log C)=\Oh(n^{r+1}C^{r+1}\log C)$ bits constitutes a polynomial kernelization with respect to~$n+C$.
\myqed
\end{proof}


\section{A polynomial kernelization for $q$-ILP Feasibility($m+C$)}\label{appendix:section:qilpfmC:pk}

\begin{proposition}\label{proposition:qilpfmC:pk}
The~$q$-\ILPF($m+C)$ problem admits a polynomial kernelization.
\end{proposition}

\begin{proof}
The main observation is that there is only a bounded number of different coefficient patterns for the variables since each variable appears only in a constant number of constraints:
\begin{enumerate}
 \item Each variable appears in some set of up to~$q$ constraints, giving less than~$(m+1)^q=\Oh(m^q)$ choices.
 \item In each such constraint, it has a coefficient of absolute value at most~$C$, giving~$(2C+1)^q$ choices.
\end{enumerate}
Now, if two variables, say~$x_i$ and~$x_j$ have the same pattern, then we delete one of them. Let us check that deletion of any single variable~$x_j$ for which we have another variable~$x_i$ with the same pattern does no harm: It is clear that a solution for the created ILP can be extended to the original by setting~$x_j=0$. Conversely, let~$x_i=a_i$ and~$x_j=a_j$ be part of a feasible point for the original ILP. For the created ILP, which does not contain~$x_j$, we create a feasible point by letting~$x_i=a_i+a_j$, and letting all other coordinates be the same. It is easy to see that due to having the same pattern, all constraints are still satisfied (this is true for any way of splitting~$a_i+a_j$ between~$x_i$ and~$x_j$).

In total we keep only one variable for each pattern, giving~$\Oh(m^qC^q)$ variables. This allows an encoding in space~$\Oh(nm\log C)=\Oh(m^{q+1}C^q\log C)$.

A technical note: We have assumed that there is no global constraint~$x\geq 0$ that we would not count as separate constraints~$x_i\geq 0$ in the parameter~$m$ (by nature of writing ILPs as, e.g.,~$Ax\leq b$,~$x\geq 0$ and taking~$m$ to be the number of rows of~$A$). If we have~$x\geq 0$ then, in the above argument, letting~$x_i=a_i+a_j$ still works (since~$a_i,a_j\geq 0$) and, similarly, setting~$x_j=0$ is feasible.
\myqed
\end{proof}


\begin{figure}[p]
\centering
\begin{tikzpicture}[scale=0.9,thick, every node/.style={rectangle split, rectangle split parts=2, draw}]
\node (nmC) at (0,8) {$n+m+C$ \nodepart{second} (PK)};
\node (nmc) at (-4,6) {$n+m+\log C$ \nodepart{second} \textbf{PK (trivial)}};
\node (nm) at (-6,4) {$n+m$ \nodepart{second} \textbf{no PK (Thm.~\ref{theorem:main:ilpfnm:nopk})}};
\node (nC) at (0,6) {$n+C$ \nodepart{second} \textbf{no PK (HS($\boldsymbol{n}$))}};
\node (nc) at (-2,4) {$n+\log C$ \nodepart{second} (no PK)};
\node (mC) at (4,6) {$m+C$ \nodepart{second} \textbf{no PK (HS($\boldsymbol{m}$))}};
\node (mc) at (2,4) {$m + \log C$ \nodepart{second} (no PK)};
\node (n) at (-4,2) {$n$ \nodepart{second} (no PK)};
\node (C) at (6,4) {$C$ \nodepart{second} \textbf{NP-hard}};
\node (m) at (0,2) {$m$ \nodepart{second} (no PK)};
\node (c) at (4,2) {$\log C$ \nodepart{second} (NP-hard)};

\draw (nmC) -- (nmc) -- (nm) -- (n);
\draw (nmC) -- (nC) -- (nc) -- (n);
\draw (nmC) -- (mC) -- (mc) -- (m);
\draw (nmc) -- (nc) -- (c);
\draw (nmc) -- (mc) -- (c);
\draw (nm) -- (m);
\draw (nC) -- (C) -- (c);
\draw (mC) -- (C);
\end{tikzpicture}
\caption{\label{figure:general:ilpf:overview} Existence of polynomial kernels for (unrestricted) \ILPF regarding parameterization by a subset of~$n$,~$m$,~$\log C$, and~$C$ (depicted in top half of nodes). ``PK'' indicates a polynomial kernel, ``no PK'' means that no polynomial kernel is possible unless \containment, and ``NP-hard'' stands for NP-hardness with constant parameter value. Positive results transfer upwards to easier (larger) parameterizations; negative results transfer downwards to harder (smaller) parameterizations. NP-hardness of~ILPF($C$) for~$C=1$ comes from \textsc{Satisfiability}.}
\end{figure}

\begin{figure}[p]
\centering
\begin{tikzpicture}[scale=0.9,thick, every node/.style={rectangle split, rectangle split parts=2, draw}]
\node (nmC) at (0,8) {$n+m+C$ \nodepart{second} (PK)};
\node (nmc) at (-4,6) {$n+m+\log C$ \nodepart{second} \textbf{PK (trivial)}};
\node (nm) at (-6,4) {$n+m$ \nodepart{second} \textbf{no PK (Cor.~\ref{corollary:ilpfnm:sparse:nopk})}};
\node (nC) at (0,6) {$n+C$ \nodepart{second} \textbf{PK (Prop.~\ref{proposition:rilpfnC:pk})}};
\node (nc) at (-2,4) {$n+\log C$ \nodepart{second} \textbf{no PK~\cite{Kratsch13_sparseilp}}};
\node (mC) at (4,6) {$m+C$ \nodepart{second} (PK)};
\node (mc) at (2,4) {$m + \log C$ \nodepart{second} \textbf{PK ($\boldsymbol{n\leq rm}$)}};
\node (n) at (-4,2) {$n$ \nodepart{second} (no PK)};
\node (C) at (6,4) {$C$ \nodepart{second} \textbf{NP-hard}};
\node (m) at (0,2) {$m$ \nodepart{second} (no PK)};
\node (c) at (4,2) {$\log C$ \nodepart{second} (NP-hard)};

\draw (nmC) -- (nmc) -- (nm) -- (n);
\draw (nmC) -- (nC) -- (nc) -- (n);
\draw (nmC) -- (mC) -- (mc) -- (m);
\draw (nmc) -- (nc) -- (c);
\draw (nmc) -- (mc) -- (c);
\draw (nm) -- (m);
\draw (nC) -- (C) -- (c);
\draw (mC) -- (C);
\end{tikzpicture}
\caption{\label{figure:rowsparse:ilpf:overview} Existence of polynomial kernels for ($r$-row-sparse)~$r$-\ILPF regarding parameterization by a subset of~$n$,~$m$,~$\log C$, and~$C$ (depicted in top half of nodes). ``PK'' indicates a polynomial kernel, ``no PK'' means that no polynomial kernel is possible unless \containment, and ``NP-hard'' stands for NP-hardness with constant parameter value. Positive results transfer upwards to easier (larger) parameterizations; negative results transfer downwards to harder (smaller) parameterizations. NP-hardness of~$r$-ILPF($C$) for~$C=1$ comes from $3$-\textsc{Satisfiability}.}
\end{figure}

\begin{figure}[p]
\centering
\begin{tikzpicture}[scale=0.9,thick, every node/.style={rectangle split, rectangle split parts=2, draw}]
\node (nmC) at (0,8) {$n+m+C$ \nodepart{second} (PK)};
\node (nmc) at (-4,6) {$n+m+\log C$ \nodepart{second} \textbf{PK (trivial)}};
\node (nm) at (-6,4) {$n+m$ \nodepart{second} \textbf{no PK (Cor.~\ref{corollary:ilpfnm:sparse:nopk})}};
\node (nC) at (0,6) {$n+C$ \nodepart{second} (PK)};
\node (nc) at (-2,4) {$n+\log C$ \nodepart{second} \textbf{PK ($\boldsymbol{m\leq qn}$)}};
\node (mC) at (4,6) {$m+C$ \nodepart{second} \textbf{PK (Prop.~\ref{proposition:qilpfmC:pk})}};
\node (mc) at (2,4) {$m + \log C$ \nodepart{second} \textbf{no PK (\cite{DomLS09}*)}};
\node (n) at (-4,2) {$n$ \nodepart{second} (no PK)};
\node (C) at (6,4) {$C$ \nodepart{second} \textbf{NP-hard}};
\node (m) at (0,2) {$m$ \nodepart{second} (no PK)};
\node (c) at (4,2) {$\log C$ \nodepart{second} (NP-hard)};

\draw (nmC) -- (nmc) -- (nm) -- (n);
\draw (nmC) -- (nC) -- (nc) -- (n);
\draw (nmC) -- (mC) -- (mc) -- (m);
\draw (nmc) -- (nc) -- (c);
\draw (nmc) -- (mc) -- (c);
\draw (nm) -- (m);
\draw (nC) -- (C) -- (c);
\draw (mC) -- (C);
\end{tikzpicture}
\caption{\label{figure:columnsparse:ilpf:overview} Existence of polynomial kernels for ($q$-column-sparse)~$q$-\ILPF regarding parameterization by a subset of~$n$,~$m$,~$\log C$, and~$C$ (depicted in top half of nodes). ``PK'' indicates a polynomial kernel, ``no PK'' means that no polynomial kernel is possible unless \containment, and ``NP-hard'' stands for NP-hardness with constant parameter value. Positive results transfer upwards to easier (larger) parameterizations; negative results transfer downwards to harder (smaller) parameterizations. The lower bound for~$q$-ILPF($m+\log C$) follows from a lower bound for \textsc{Small Subset Sum}($k+d$) due to Dom et al.~\cite{DomLS09}: Expressing \textsc{Subset Sum} with weights of value at most~$2^d$ requires only~$m=\Oh(1)$ constraints and the encoding size~$\log C$ per weight of value at most~$C=2^d$ is~$\Oh(d)$. NP-hardness of~$q$-ILPF($C$) for~$C=1$ comes from \textsc{Satisfiability} restricted to instances with variable degree at most three (cf.~\cite{GareyJ1979}).}
\end{figure}

\end{document}